\documentclass[aps,pra,twocolumn,amsmath,amssymb,nofootinbib,showpacs,superscriptaddress]{revtex4-2}
\usepackage[english]{babel}
\usepackage{latexsym}
\usepackage{graphicx}
\usepackage{epstopdf}
\usepackage{epsfig}
\usepackage{color}
\usepackage{bm}
\usepackage{amsmath}
\usepackage{amssymb}
\usepackage{amsthm}
\usepackage{dcolumn}
\usepackage{float}
\usepackage{hyperref}
\usepackage{color}
\usepackage{cleveref}
\usepackage[svgnames,dvipsnames]{xcolor}
\usepackage{enumerate}
\usepackage{braket}

\hypersetup{hidelinks,colorlinks=true,allcolors=DarkBlue}

\usepackage{algorithm}
\usepackage{algpseudocode}
\algrenewcommand\algorithmicrequire{\textbf{Input:}}
\algrenewcommand\algorithmicensure{\textbf{Output:}}

\usepackage[svgnames,dvipsnames]{xcolor}

\usepackage{tikz}
\usetikzlibrary{positioning}

\newtheorem{theorem}{Theorem}
\newtheorem{lemma}{Lemma}
\newtheorem{proposition}{Proposition}
\newtheorem{corollary}{Corollary}
\newtheorem{observation}{Observation}

\theoremstyle{remark}
\newtheorem{remark}{Remark}

\begin{document} 

\title{Spanning-tree-packing protocol for conference key propagation in quantum networks}

\author{Anton Trushechkin}
\email{anton.trushechkin@hhu.de}

\affiliation{Heinrich Heine University D\"{u}sseldorf,
Faculty of Mathematics and Natural Sciences, Institute for Theoretical Physics III, 
Universit\"{a}tsstr. 1, D\"{u}sseldorf 40225, Germany}

\affiliation{Steklov Mathematical Institute of Russian Academy of Sciences, Gubkina 8, Moscow 119991, Russia}

\author{Hermann Kampermann}

\affiliation{Heinrich Heine University D\"{u}sseldorf,
Faculty of Mathematics and Natural Sciences, Institute for Theoretical Physics III, 
Universit\"{a}tsstr. 1, D\"{u}sseldorf 40225, Germany}

\author{Dagmar Bru\ss}
\affiliation{Heinrich Heine University D\"{u}sseldorf,
Faculty of Mathematics and Natural Sciences, Institute for Theoretical Physics III, 
Universit\"{a}tsstr. 1, D\"{u}sseldorf 40225, Germany}

\date{\today}

\begin{abstract}

We consider a network of users connected by pairwise quantum key distribution (QKD) links. Using these pairwise secret keys and public classical communication, the users want to generate a common (conference) secret key at the maximal rate. We propose an algorithm based on spanning-tree packing (a known problem in graph theory) and prove its optimality. This algorithm enables optimal conference key generation in modern quantum networks of arbitrary topology. Additionally, we discuss how it can guide the optimal placement of new bipartite links in the network design.

\end{abstract}

\maketitle

\section{Introduction} 

Quantum key distribution (QKD) is able to provide secure communication \cite{BB84,Gisin,Scarani,PirandolaRev}. This is key in view of the progress in quantum computation observed now \cite{FedorovRev,MartinisRev2024}, but also in view of possible inventions in the field of efficient algorithms that can threaten conventional cryptography. Originally, QKD protocols allow two users to establish a common secret key (bitstring). 
However, the real world is often interested not only in point-to-point communication, but in the communication of many users -- network communication. For the widespread adoption of quantum communication technologies, a well-developed theory of quantum networks is essential.

Several theoretical and experimental works have investigated networks of nodes connected by sources of bipartite or multipartite entangled states  \cite{Epping2016n1,Epping2016n2,Wallnofer2019,DeBone2020,Vicente2021,Sukachev,Avis2023,Gupta2023,Sutcliffe2023,PENexper,Wehner2024}. Multipartite QKD, or quantum conference key agreement, which uses multipartite entangled states has been suggested and analyzed in Refs.~\cite{QCKAreview,Epping2017,SW2023,Pirandola2020,Das2021,Carrara2021,Fedrizzi,CVmultiQKD2024}. In the future, such networks are hoped to constitute a ``quantum internet'' \cite{Kimble,QIntenet-Pirandola,QIntenet-Simon,QIntenet-Wehner,QInternet-book,QInternet2025,QPhotNet2025,QInternet2025Tech}. 

However, in the near future one would expect technologically simpler networks of nodes connected by pairwise QKD, not requiring the creation and manipulation of multipartite entangled states. Such bipartite QKD networks already exist in a number of countries \cite{EvolQKDnet}. Protocols for classical information transmission in QKD networks are under active development \cite{KozubovNet,KikTayFed,TrustedNodePQC} and the general structure of future QKD networks is being investigated \cite{DemoQuanDT,QKDcomplnet}.

Examples of graphs representing such networks are given in Fig.~\ref{FigStarTri}. Vertices correspond to users and an edge between two nodes means that the corresponding pair of users is connected by a point-to-point (bipartite) QKD link. We assume that all QKD links can operate in parallel, although alternative scenarios may also be considered.

A variety of tasks can be studied for such networks. In this work, we focus on the task of generating a common (conference) secret key. It is well known (see, e.g., \cite{Epping2017}) that  users having a connected network of bipartite keys can agree on a common (conference) secret key using classical communication (XOR operations). The question we address is to do it optimally, i.e., to generate the conference key at the maximal rate given the bipartite rates in the network. We refer to this task as optimal \textit{conference key propagation} because we want to propagate pairwise secret keys over the network, resulting in a conference key. 

\begin{figure}
\centering

\includegraphics[scale=1]{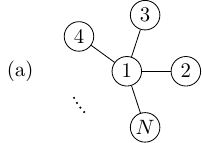}
\qquad\qquad
\includegraphics[scale=1]{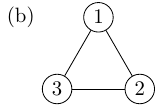}
\vspace{6mm}

\includegraphics[scale=1]{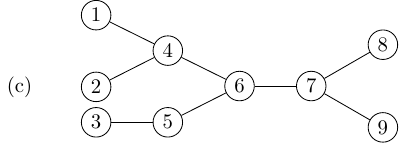}

\caption{Examples of QKD networks: a star network with $N$ users (a), the triangle network (b), and a tree network (c). Vertices (nodes) of a graph represent users and edges represent (bipartite) QKD links between them.}
\label{FigStarTri}
\end{figure}

We propose a protocol of conference key propagation for networks of arbitrary topologies based on the spanning-tree packing (STP). Optimal spanning-tree packing is a known and efficiently solvable problem in graph theory \cite{Diestel,SpanSurvey,Barahona1995}. We prove that the optimal spanning-tree packing gives the optimal conference key propagation in this setting and
can be used in existing QKD networks.

The corresponding results are known in classical multiterminal information theory \cite{PIN2010}. We adapt them to the context of quantum key distribution and quantum networks and consider additional applications in detail. Namely, we discuss the derivation of simple upper bounds for the conference key rate for networks with bottleneck structures. We also explore how this framework can guide the optimal placement of new QKD links in network design.

Note that other approaches to quantum conference key agreement without direct use of multipartite entangled states include the MDI (measurement-device-independent) approach, where the multipartite entangled state is postselected \textit{a posteriori} in a system of detectors \cite{MDIQCKA2025,TFQKD2019,TFQKD2022,TFQKD2022Lo,TFQKD2023,FullyPassiveQCKA,PMQCC,MDI3users,TFQKD2025,MDIQCKASideChSecure}, and simultaneous delivering of the copies of the same quantum state from a sender to many receivers \cite{MultiCVQKD}. 

Our paper is organized as follows. In Sec.~\ref{SecNot}, we state the problem. In Sec.~\ref{SecDescr}, we describe the protocol starting from simple examples. In Sec.~\ref{SecOpt}, we formulate a theorem of its optimality. In Sec.~\ref{SecApp}, we discuss the upper bounds on the conference key rate originating from complex bottleneck structures and a related problem of network optimization. In Sec.~\ref{SecSecurityParam}, we give the formula for the universally composable security parameter $\varepsilon_{\rm conf}$ for the conference key generated by the proposed STP protocol. 
Since the existing efficient algorithm solving the spanning-tree packing contains many steps \cite{Barahona1995}, in Appendix~\ref{SecAlg} we propose a simpler heuristic algorithm for this problem used in the examples of this paper.
In Appendix~\ref{SecProofOpt}, we give a proof of the main theorem.

\section{Problem statement}
\label{SecNot}

A network is represented as a graph $(\mathcal V,\mathcal E)$, where $\mathcal V=\{1,\ldots,N\}$ is a set of vertices (nodes) associated with users and $\mathcal E$ is the set of edges associated with bipartite QKD links between the users. We always assume that the graph is connected, i.e., there is a path between any pair of vertices.

We assume that all QKD links $e\in\mathcal E$ operate in parallel for time $T$ and generate the keys $K_e$ of lengths $L_e$. Formally, $K_e$ are random variables taking values on $\{0,1\}^{L_e}$. Each launch of all QKD links for time $T$ will be called ``a round.'' Then, the bipartite secret key rates are $r_e=L_e/T$. If after $n$ rounds (which take the total time $nT$), we generate a conference key of the length $L_{\rm conf}$, then the conference key rate is given by $r_{\rm conf}=L_{\rm conf}/(nT)$. 

In practice, the bipartite secret key rates are typically integers (e.g., certain number of bits per second), which will be assumed for simplicity. 
However, some of our formulas and calculations do not rely on the property of $r_e$ being integer and are valid also for rational or real $r_e$.  

The task is stated as follows: given a graph $(\mathcal V,\mathcal E)$ and bipartite rates $\{r_e\}_{e\in\mathcal E}$, build a conference key propagation protocol with the maximally possible rate $r_{\rm conf}$. For a formal definition, see Eq.~(\ref{EqConfCapacity}) in Sec.~\ref{SecOpt}.

Thus, we deal with a weighted graph with the edge weights given by the bipartite key rates $r_e$, see Fig.~\ref{FigMulti}(a). Since the weights $r_e$ are integers, we can treat the graph as a multigraph (i.e., graphs where multiple edges between the same pair of vertices are allowed)  with the multiplicity of edges $r_e$, see Fig.~\ref{FigMulti}(b). 

\begin{figure}
\centering

\includegraphics[scale=1]{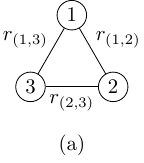}
    {}
    \hspace{10mm}
    {}
\includegraphics[scale=1]{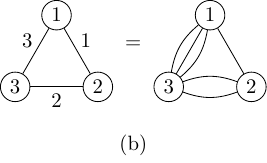}

\caption{Quantum QKD network as a weighted graph with bipartite secret key rates $r_e$ as the weights (a). In the case of integer weights, it can be alternatively represented as the multigraph where edge multiplicities correspond to the weights (b).}
\label{FigMulti}
\end{figure}

Note that, as usual in QKD, we neglect the influence of classical communication and computation times on the key rate because they are several orders of magnitude smaller than the rates of distribution of quantum states, especially over large distances. Indeed, classical signals with macroscopic numbers of photons are much less vulnerable to losses than quantum single-photon states or weak coherent pulses. Thus, the conference key rate is limited mainly by the bipartite key rates $r_e$ and the network topology, while the influence of classical communication and computation times is negligible.

Also we use the standard assumptions in QKD: the parties are trusted and the classical communication channels are public, but authenticated. The latter assumption means  that the classical channels can be freely eavesdropped, but the adversary cannot change the transmitted messages or send their own ones. In practice, this is achieved, e.g., using digital signatures or hash-based authentication.

It is known that keys generated by QKD deviate from being perfectly secure and this deviation is characterized by a security parameter $\varepsilon$ \cite{Portmann2014,Portmann2022}. So, in general, each edge $e$ is  characterized by a security parameter $\varepsilon_e$. We start with perfectly secure bipartite keys and return to the practical case of nonzero $\varepsilon_e$ in Sec.~\ref{SecSecurityParam}.

\section{Description of the protocol}
\label{SecDescr}

\subsection{Simple case: tree networks}
\label{SecTreeNet}

Before we formulate the general protocol, let us start with simple examples. Consider the star network depicted in Fig.~\ref{FigStarTri}(a), where the central node~1 is connected to $N-1$ other nodes $2,\ldots,N$ and there are no more connections in the network. Thus, $\mathcal E=\{(1,2),(1,3),\ldots,(1,N)\}$. In such networks, user~1 is distinguished and often called ``Alice,'' while all others are ``Bobs.''

Also, for simplicity, let us assume that all bipartite secret key rates are $r_{1j}=1$, $j=2,\ldots,N$, i.e., each QKD link generates one bipartite secret bit per round. The conference key is generated as follows, see, e.g., Ref.~\cite{Epping2017}. The participants publicly choose which of the bipartite keys will be the conference key. For example, they choose the key $K_{12}$ as the future conference key. Then Alice encrypts $K_{12}$ as the message to the Bobs $3,\ldots,N$ by the one-time pad encryptions using the corresponding bipartite keys $K_{(1,3)}, K_{(1,4)},\ldots, K_{(1,N)}$. Namely, Alice publicly announces the bitstrings 
\begin{equation}
\label{EqSimpleXORs}
\begin{split}
&K_{(1,2)}\oplus K_{(1,3)},
\\
&K_{(1,2)}\oplus K_{(1,4)},\\ 
&\ldots\\
&K_{(1,2)}\oplus K_{(1,N)}.
\end{split}
\end{equation} 
Here $\oplus$ is XOR (addition modulo~2). The users $3,4,\ldots,N$ can decrypt the ciphertext (apply the operation $\oplus K_{(1,j)}$ for the corresponding $j$ again) to obtain the bit $K_{(1,2)}$, which then can be used as a conference key. Since we have generated one bit of the conference key with one round of bipartite QKD protocols,  $r_{\rm conf}=1$. This is a known obvious solution.

If $r_e$ are different integers, then the conference key rate is upper bounded by the minimal bipartite rate $\min r_e$. Keeping $\min r_e$ bits in each of the keys and repeating the described procedure gives $r_{\rm conf}=\min r_e$. 

This protocol can be easily generalized to any connected graph without cycles, i.e., a tree. Let us again assume that all bipartite secret key rates are $r_e=1$. The protocol is illustrated in Fig.~\ref{FigTree} for the graph given in Fig.~\ref{FigStarTri}(c). Again, the participants publicly agree on which of the bipartite keys $K_{\bar e}$ will be the conference key. Since there are no cycles in the graph, removal of the edge $\bar e=(i,j)$ from the graph breaks the tree graph into two nonconnected subtrees. In Fig.~\ref{FigTree}, the removal of the edge $(6,7)$ breaks the graph into two subtrees depicted by blue and green with vertices~6 and~7 as the roots. Setting $i$ and $j$ as the roots of these two subtrees, we can make the graph oriented and thus to construct unique paths from $i$ or $j$ to every vertex of the corresponding subtree. 

In our example, vertex 6 encrypts the message $K_{\bar e}$ for its children in the corresponding subtree (vertices 4 and 5) using the one-time pad and the corresponding bipartite secret keys. Then the vertices that are not leafs in the subtrees do the same for their children in the subtrees. In our example, vertices 4 and 5 encrypt the future conference key $K_{(5,6)}$ for the users 1, 2, and 3 using the bipartite keys $K_{(1,4)}$, $K_{(2,4)}$, and $K_{(3,5)}$. Vertex 7 and its descendants act in an analogous way. Thus, the conference key $K_{(5,6)}$ propagates along the whole network.

In this version of the protocol, vertices 4 and 5 have to wait for messages from vertex 6 in order to send their messages for vertices 1, 2, and 3. Let us give an equivalent protocol, which is single round, i.e., each participant makes public announcements independently on other announcements. After all announcements have been made, each participant can recover the conference key. 

Namely, each vertex $\alpha$ except the roots $i$ and $j$ of the subtrees (6 and 7 in our example) has exactly one incoming (directed) edge ${\rm in}_\alpha$ and a (possibly empty) set of outcoming edges ${\rm out}_\alpha$. Also, by definition, we put ${\rm in}_i=j$ and ${\rm in}_j=i$. Then each user $k$ with nonempty ${\rm out}_k$ publicly announces the bitstrings 
\begin{equation}\label{EqTreeAnnounce}
C_{e}^{(\alpha)}=K_{{\rm in}_\alpha}\oplus K_{e} 
\end{equation}
for all $e\in{\rm out}_\alpha$. That is, each user encrypts the bitstring $K_{{\rm in}_\alpha}$ using one-time pad and keys $K_e$, $e\in{\rm out}_\alpha$, and announces the ciphertexts. Then, each user, step by step, going back from its vertex to $i$ and $j$, decrypt the ciphertexts and finally recovers the bitstring $K_{\bar e}=K_{(i,j)}$, which can be used as a conference key.

\begin{figure}
\centering
\includegraphics[scale=1]{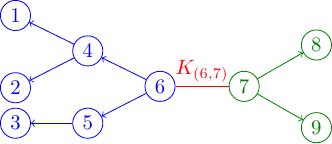}
\caption{Conference key propagation for the tree graph from Fig.~\ref{FigStarTri}(c). The users agree whose bipartite key will be the conference key. Here: $K_{(6,7)}$. Removal of the corresponding edge breaks the graph into two non-connected subgraphs (subtrees, depicted by the blue and green colors) with vertices 6 and 7 as the roots. The propagation of the key $K_{(6,7)}$ proceeds according to the arrows.}
\label{FigTree}
\end{figure}

In our example, the announcements are:
\begin{gather*}
C_{(4,6)}^{(6)}=K_{(6,7)}\oplus K_{(4,6)},
\qquad
C_{(5,6)}^{(6)}=K_{(6,7)}\oplus K_{(5,6)},
\\
C_{(1,4)}^{(4)}=K_{(4,6)}\oplus K_{(1,4)},
\qquad
C_{(2,4)}^{(4)}=K_{(4,6)}\oplus K_{(2,4)},
\\
C_{(3,5)}^{(5)}=K_{(5,6)}\oplus K_{(3,5)},
\\
C_{(7,8)}^{(7)}=K_{(6,7)}\oplus K_{(7,8)},
\qquad
C_{(7,9)}^{(7)}=K_{(6,7)}\oplus K_{(7,9)}.
\end{gather*}
Then, for example, user~1 recovers the conference key rate calculating
\begin{equation*}
K_{(1,4)}\oplus C_{(1,4)}^{(4)}\oplus C_{(4,6)}^{(6)}=K_{(6,7)}.
\end{equation*}
As in the case of the star graph, $L_{\rm conf}=1$ and $r_{\rm conf}=1$. If the bipartite keys are perfectly secure, then, due to perfect security of the one-time pad and the chainlike structure of encryptions (\ref{EqTreeAnnounce}), the conference key is also perfectly secure. For arbitrary bipartite keys $r_e$, we can apply the same reasonings as for the star graph and conclude that $r_{\rm conf}=\min r_e$.

Thus,  we arrive at the following observation:  
\begin{observation}
\label{ObsTree}
A tree of bipartite secret bits leads to one conference secret bit.
\end{observation} 
For optimal conference key propagation, a general problem is to pack as many trees into a given graph as possible, which is a well-known problem in graph theory. Before we give precise formulations of this problem in Sec.~\ref{SecSTP}, we consider one more instructive simple example in the next subsection.

\subsection{Triangle network}

Now consider the simplest network with a cycle: a triangle network depicted in Fig~\ref{FigStarTri}(b). For simplicity, we still assume that all bipartite key rates are equal to one. How can the users agree on a conference secret key? Of course, they can ignore, for example, the QKD link $(2,3)$ and reduce the problem to the star graph case described in Sec.~\ref{SecTreeNet}, which leads to the conference key rate $r_{\rm  conf}=1$. In general, every graph with cycles can be reduced to a graph without cycles by removing a subset of edges. However, ignoring QKD links might be nonoptimal.

Let us propose an improved protocol. Let us launch two rounds of bipartite QKD links. Then, each link generates two secret bits $K_{e,1}$ and $K_{e,2}$. Then the participants 1,2 and 3 announce, respectively,
\begin{equation}
\begin{split}
C^{(1)}&=K_{(1,2),1}\oplus K_{(1,3),1},\\ 
C^{(2)}&=K_{(1,2),2}\oplus K_{(2,3),1},\\
C^{(3)}&=K_{(1,3),2}\oplus K_{(2,3),2}.
\end{split}
\end{equation}
This information is sufficient for all participants to recover all bipartite keys. Then, they can decide to accept, e.g., $K_{(1,2),1}$, $K_{(1,2),2}$, and $K_{(1,3),2}$ as three bits for the conference key. That is, we have generated $L_{\rm conf}=3$ bits of the conference key using $n=2$ rounds of bipartite links, hence, the conference key rate is
\begin{equation}
r_{\rm conf}=\frac{L_{\rm conf}}{n}=\frac{3}2,
\end{equation}
which is greater than the value $r_{\rm  conf}=1$ obtained by ignoring one of the links. Since the eavesdropper does not know $K_{(1,3),1}$, $K_{(2,3),1}$, and $K_{(2,3),2}$, she has no information about the conference key. Thus, we have obtained a better protocol for conference key propagation, which uses all the bipartite QKD links. In the following, we generalize and formalize the scheme given in this example.

\subsection{General formulation of the STP-based conference key propagation protocol}
\label{SecSTP}

The essence of the protocol for the last example is depicted in Fig~\ref{FigPack} (top row). Consider again the graph $(\mathcal V,\mathcal E)$ with the weights (bipartite key rates) $r_e$, or, equivalently, the corresponding multigraph with the edge multiplicities $r_e$ [see Section~\ref{SecNot} and Fig.~\ref{FigMulti}(b)]. Consider now the sequence of multigraphs $(\mathcal V,\mathcal E,\{nr_e\}_{e\in\mathcal E})$, $n=1,2,\ldots$, with the edge multiplicities $nr_e$ (numbers of bipartite bits generated in $n$ rounds).
Recall that a subgraph of a graph or a multigraph is called a spanning tree if it is a tree (i.e., connected graph without cycles) and contains all vertices of the original graph. We can reformulate Observation~\ref{ObsTree} as follows:
\begin{observation}
A spanning tree of bipartite secret bits in the multigraph $(\mathcal V,\mathcal E,\{nr_e\})$ leads to one conference secret bit.
\end{observation} 

If we wish to maximize the conference key rate, we need to pack as many spanning trees into the multigraph $(\mathcal V,\mathcal E,\{nr_e\})$ in the following sense:

\begin{equation}
\label{EqRateMultigraph}
    r_{\rm conf}=
    \max_{k,n}
    \Big\{
    \frac kn\,\Big|\,
    \begin{smallmatrix}
    (\mathcal V,\mathcal E,\{nr_e\}) \text{ contains $k$ edge-disjoint}\\ \text{spanning trees}
    \end{smallmatrix}
    \Big\}.
\end{equation}
That is, if we generated the conference key of length $k$ using $n$ rounds of parallel bipartite QKD links, then the conference key rate is $k/n$ and we want to maximize this quantity. In our example with the triangle network depicted in Fig~\ref{FigPack} (top row), the multigraph $(\mathcal V,\mathcal E,\{2r_e\})$ can be partitioned into three edge-disjoint spanning trees, hence the conference key rate is 3/2. Fig.~\ref{FigPack} shows the spanning-tree packing for other graphs also under the assumption of all the bipartite secret key rates are equal to one.

\begin{figure*}
%\centering
\includegraphics[scale=1]{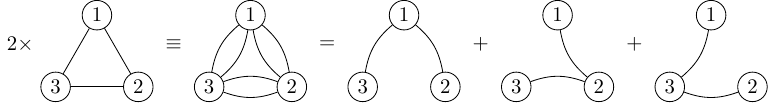}

\vspace{10mm}

\includegraphics[scale=1]{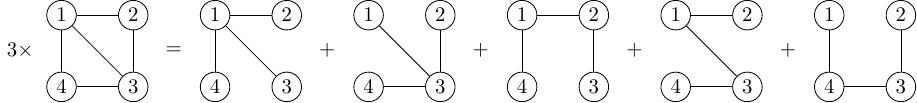}

\vspace{10mm}
\includegraphics[scale=1]{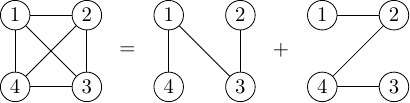}

\vspace{10mm}
\includegraphics[scale=1]{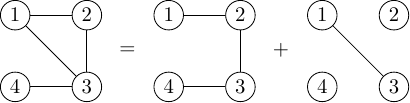}

\caption{Spanning-trees packings in different graphs with equal edge weights (bipartite QKD rates) for all edges (note that this assumption is used only for illustrative purposes in the figures, while the general theory is formulated for arbitrary bipartite key rates $r_e$, see the text). Here it is convenient to think that an edge in a graph represents a bipartite key with a fixed length $L$. Then multiple QKD sessions (factors in front of of the graphs) give multiple edges between the same vertices. One spanning tree correspond to a conference key of the length $L$. Thus, the problem is to find an optimal \textit{spanning-tree packing}, i.e. a set of edge-disjoint spanning trees such that the ratio of the number of edge-disjoint spanning trees to the number of QKD sessions (per each bipartite link) is maximal. The last term in the last line is not a spanning tree, but a leftover edge. The operations $\times$ and $+$ are understood here in the sense of edge multiplicities in a multigraph. If all bipartite secret key rates are equal to one ($L=1$), then the conference key rates $r_{\rm conf}$ for these examples (from top to bottom) are 3/2, 5/3, 2, and 1. We found the spanning-tree packings here using the heuristic described in Appendix~\ref{SecAlg}.
}\label{FigPack}
\end{figure*}

The problem  of finding the maximal number of edge-disjoint spanning trees in a graph or multigraph is called the \textit{spanning-tree-packing problem} and is well-known in graph theory \cite{Diestel,SpanSurvey,Barahona1995}. Its solution, i.e., the maximum itself is  called the \textit{spanning-tree-packing number}. More precisely, the standard formulation of the spanning-tree packing assumes that the multigraph is fixed, which corresponds to the case when the maximization in Eq.~(\ref{EqRateMultigraph}) is performed only over $k$ for a fixed $n$. Since, in our case, we are free to choose the number of bipartite rounds for the generation of the conference key, we use a slightly different formulation.

The security of the conference key obtained by spanning-tree packing relies on the security of the one-time pad and the requirement that spanning trees must be edge disjoint, hence, each bipartite key is used only once.

Of note, an efficient algorithm of finding optimal spanning-tree packing exists \cite{Barahona1995}. Also, in Appendix~\ref{SecAlg}, we give a simple heuristic algorithm. Formally, in contrast to the algorithm from Ref.~\cite{Barahona1995}, it has at least exponential time complexity. Nevertheless, in practice, it allows one to easily find optimal spanning-tree packings for small graphs or graphs with recognizable structures (certain patterns of edges, recognizable clusters, etc.) and was used for all examples in this paper.

The spanning-tree-packing number is given by the Nash-Williams--Tutte theorem \cite{Diestel,Barahona1995,SpanSurvey}. Consider a vertex partition $P=\{\mathcal V_1,\ldots,\mathcal V_p\}$, so that $\mathcal V=\mathcal V_1\sqcup\ldots\sqcup\mathcal V_p$ (recall that the symbol $\sqcup$ denotes the union of disjoint sets), $p\geq2$, and all subsets $\mathcal V_i\subset \mathcal V$ are nonempty. Denote $\mathcal E(P)$ the set of cross edges in the graph $(\mathcal V,\mathcal E)$, i.e., the edges whose vertices belong to different subsets in the partition $P$. Also, according to standard mathematical notations, $|P|=p$ is the number of elements (vertex subsets) in the partition $P$. Then, the multigraph $(\mathcal V,\mathcal E,\{nr_e\})$ has exactly
{}
    \begin{equation}
    \label{EqNashTutten}
        k\equiv L_{\rm conf}^{(n)}
       =
          \min_{
        \begin{smallmatrix}
            \text{vertex}\\
            \text{partitions } P  
        \end{smallmatrix}
        }    
        \bigg\lfloor
        \frac{1}{|P|-1}
        \sum_{e\in\mathcal E(P)} nr_e
        \bigg\rfloor
    \end{equation}
    edge-disjoint spanning trees, which correspond to the conference key length $L^{(n)}_{\rm conf}$ and the conference key rate
    $r_{\rm conf}^{(n)}
        =
        L_{\rm conf}^{(n)}/n.$
For large enough $n$, the argument of the floor function $\lfloor\cdot\rfloor$  is integer and, thus, this function can be removed. This leads to the conference key rate 
{}
\begin{equation}
\label{EqNashTutte}
    r_{\rm conf}
    =
    \min_{
    \begin{smallmatrix}
        \text{vertex}\\
        \text{partitions } P  
    \end{smallmatrix}
    }    
    \frac{1}{|P|-1}
    \sum_{e\in\mathcal E(P)}r_e.
\end{equation}
This is indeed the solution of the optimization problem (\ref{EqRateMultigraph}): as we argued, the rate given by the right-hand side of
Eq.~(\ref{EqNashTutte}) is achieved by choosing a proper $n$ and the rate cannot be greater than the right-hand side of Eq.~(\ref{EqNashTutte}) because $r_{\rm conf}^{(n)}\leq r_{\rm conf}$ for all $n$. In this paper, we will refer to Eq.~(\ref{EqNashTutte}) as the Nash-Williams--Tutte formula.

\section{Optimality of the STP-based protocol}
\label{SecOpt}

We have considered the problem of the optimal spanning-tree packing. Can we propose another conference key propagation protocol (not necessarily related to the spanning-tree packing) for a  network of bipartite QKD links that would give better results than Eqs.~(\ref{EqRateMultigraph}) and (\ref{EqNashTutte})? The answer is negative: optimal spanning-tree packing is also optimal among all possible conference key propagation protocols. In order to prove it rigorously, we need to give a formal problem statement.

The conference key propagation problem is a particular case of secret key distillation in classical networks \cite{Csiszar2004}. Let, again, $K_e$, $e\in\mathcal E$, be secret keys, i.e., random variables uniformly distributed on the sets $\{0,1\}^{r_e}$. 
Denote by $\mathcal E_i$ the set of edges incident to the vertex $i\in\{1,\ldots,N\}$ and by 
\begin{equation}
\label{EqXK}
X_i=(K_e)_{e\in\mathcal E_i}
\end{equation}
the random variable obtained by the user $i$, i.e., the collection of bipartite keys of all its connections. The participant starts from $n$ independent and identically distributed (iid) repetitions of $K_e$ and, thus, $X_i$. Denote $n$ iid repetitions of $X_i$ as $X_i^n$. Each participant $i$ generates also local randomness, i.e., a random variable $R_i$ from an arbitrary (finite) alphabet. Then the participants publicly announce (classical) messages $C_1,\ldots,C_M$, where each $C_j$ is a function of  $(X_i^n,R_i)$ for $i = (j-1\bmod N)+1$ and of the previous messages $C_1,\ldots,C_{j-1}$. Denote $C=(C_1,\ldots,C_M)$. Finally, each participant $i$ calculates its version of the conference key $K_{\rm conf}^{(i)}\in\{0,1\}^m$ as a function of $(X_i^N,R_i,C)$. Actually, this is a protocol with classical local operations and public communication (CLOPC) with a finite number of communication rounds. Following the notations of Ref.~\cite{LOCC}, denote it as ${\rm CLOPC}_{\mathbb N}$, where the subscript means that the number of rounds can be arbitrary large, but finite. 

Denote by $P_{X_1\ldots X_N}$ and $P_{K_{\rm conf}^{(1)}\ldots K_{\rm conf}^{(N)}C}$ the original distribution of $X_1,\ldots,X_N$ and the final distribution of the version of the conference key and the communication, respectively. Then
\begin{equation}
\Lambda(P_{X_1\ldots X_N}^{\otimes n})=
P_{K_{\rm conf}^{(1)}\ldots K_{\rm conf}^{(N)}C}
\end{equation}
for some $\Lambda\in{\rm CLOPC}_{\mathbb N}$. Denote also by $P_C$ the marginal distribution of $C$. The \textit{conference key propagation capacity} is defined as

\begin{multline}
\label{EqConfCapacity}
\overline r_{\rm conf}=\lim_{\varepsilon\to0}
\sup_{
\begin{smallmatrix}
n,m,\\
\Lambda\in{{\rm CLOPC}}_{\mathbb N}
\end{smallmatrix}
}\Big\{\frac mn\Big|\:\frac12\|\Lambda(P_{X_1\ldots X_N}^{\otimes n})
\\-P_{K^{(1)}_{\rm conf}\ldots K^{(N)}_{\rm conf}}^{(m), \,{\rm ideal}}\otimes P_C\|_1\leq\varepsilon\Big\},
\end{multline}
where
\begin{equation}
P_{K^{(1)}_{\rm conf}\ldots K^{(N)}_{\rm conf}}^{(m), \,{\rm ideal}}
(k_1,\ldots,k_N)
=\begin{cases}
2^{-m},&k_1=\ldots=k_N,\\
0,&\text{otherwise}
\end{cases}
\end{equation}
and $\frac12\|\cdot\|_1$ denotes the total variational distance of two probability distributions. Definition (\ref{EqConfCapacity}) is given for a general multipartite probability distribution $P_{X_1\ldots X_N}$, but the special form (\ref{EqXK}) of $X_i$ corresponds to the bipartite network structure. 

Now we can formulate our main theorem.

\begin{theorem}
\label{ThOpt}
The conference key propagation capacity (\ref{EqConfCapacity}) for an arbitrary bipartite network structure (\ref{EqXK}) is equal to the spanning-tree-packing number (\ref{EqNashTutte}).
\end{theorem}

The proof is given in Appendix~\ref{SecProofOpt}. It is based on combining the theorem by Csisz\'{a}r and Narayan \cite{Csiszar2004} about the conference secret key capacity (for classical networks) with the Nash-Williams--Tutte formula (\ref{EqNashTutte}). The main technical ingredient is the following proposition:

\begin{proposition}
\label{ThLP}
    Let a graph $(\mathcal V,\mathcal E)$ with $N$ vertices and non-negative real edge weights $\{r_e\}_{e\in\mathcal E}$ be given. The number $r_{\rm conf}$ given by Eq.~(\ref{EqNashTutte}) is equal to the number $Z$ obtained from the following linear program:
    \begin{subequations}
    \label{EqLP}
    \begin{eqnarray}
    &&Z=\sum_{e\in\mathcal E}r_e-R_{\rm CO},
    \label{EqLPC}
\\
    &&R_{\rm CO}=
    \min_{R_1,\ldots,R_N}
    \sum_{i=1}^N 
    R_i 
    \quad
    \text{such that}
    \\
    &&\sum_{i\in I}
    R_i
    \geq
    \sum_{e\in\mathcal E(I)}
    r_{e},
    \quad
    \forall I\subsetneq[N].
    \label{EqLPCons}
\end{eqnarray}
\end{subequations}
Here $[N]=\{1,\ldots,N\}$ and $\mathcal E(I)$ denotes the set of edges of the subgraph induced by the subset of vertices $I$.
\end{proposition}

The proof is also given in Appendix~\ref{SecProofOpt}. Note that Proposition~\ref{ThLP} deals with the formula (\ref{EqNashTutte}) and the linear program (\ref{EqLP}), which are well defined for real (not necessarily integer) $r_e$.

Let us give an example of linear program (\ref{EqLP}) for the triangle network with weights from Fig.~\ref{FigMulti}(a):
    \begin{equation}
    \begin{split}
    &Z=r_{(1,2)}+r_{(1,3)}+r_{(2,3)}-R_{\rm CO},
%    \label{EqLPC}
\\
    &R_{\rm CO}=
    \min_{R_1,R_2,R_3}(R_1+R_2+R_3)
    \quad
    \text{such that}
    \\
    &
    R_1+R_2\geq r_{(1,2)},\\
    &R_1+R_3\geq r_{(1,3)},\\
    &
    R_2+R_3\geq r_{(2,3)},\\
    &R_1\geq0,\:R_2\geq0,\:
    R_3\geq0.
%    \label{EqLPCons}
\end{split}
\end{equation}
Its solution is
\begin{equation}
\label{EqTriSol}
    Z=
    \begin{cases}
        r_{(1,2)}+r_{(1,3)},
        &\text{if }r_{(1,2)}+r_{(1,3)}\leq r_{(2,3)},\\
        r_{(1,2)}+r_{(2,3)},
        &\text{if }r_{(1,2)}+r_{(2,3)}\leq r_{(1,3)},\\
        r_{(1,3)}+r_{(2,3)},
        &\text{if }r_{(1,3)}+r_{(2,3)}\leq r_{(1,2)},\\
        \frac12[r_{(1,2)}+r_{(1,3)}+r_{(2,3)}], &\text{otherwise.}    
    \end{cases}
\end{equation}
That is, if there exists a bipartite rate that is larger than the sum of the two others, then the conference key rate is limited by the sum of the two smaller rates. Such bottleneck structures will be considered in more detail in Sec.~\ref{SecUpBnd}. Otherwise, the conference key rate is one half of the sum of all bipartite rates. If all three bipartite rates are equal to one, then we recover the result $r_{\rm conf}=Z=3/2$ we obtained before.

The linear program (\ref{EqLP}) results from Csisz\'{a}r and Narayan's theorem applied to our case (a general formulation is given in Sec.~\ref{SecReduction} in Appendix~\ref{SecProofOpt}). The first term in the right-hand side of Eq.~(\ref{EqLPC}) is the total amount of (bipartite) randomness in the network generated per round. In the second term, $R_i$ is the amount of bits announced by the $i$th participant per round such that each participant can recover the private randomness of all other participants (``omniscience''). Then $R_{\rm CO}$  denotes the minimal total amount of announced information for omniscience (``communication for omniscience''). Each constraint (\ref{EqLPCons}) reflects the fact that the participants from $[N]\backslash I$ (if we consider them as a whole) are ignorant only about the bipartite bits that are ``internal'' for the subgraph induced by the vertices $I$. Intuitively, this constraint can be understood as follows: the total amount of information announced by the participants from $I$ cannot be smaller than the total amount of randomness in their subnetwork. In summary, the maximally possible conference key rate is the total amount of randomness in the network minus the amount of publicly revealed information about this randomness.

Thus, Proposition~\ref{ThLP} relates the linear program (\ref{EqLP}) originating from information-theoretic considerations with the graph-theoretic formulas (\ref{EqRateMultigraph}) and (\ref{EqNashTutte}).

\section{Further applications of the \\STP protocol}
\label{SecApp}

The main application of the proposed spanning-tree-packing protocol is the optimal conference key propagation in a network of bipartite QKD links. In this section, we discuss two further applications: upper bounds on the conference key propagation from graph topology and bottleneck structures (subsection~\ref{SecUpBnd}), and network optimization, namely, optimal allocation of new bipartite QKD links in the network design (subsection~\ref{SecNetOpt}). 

\subsection{Upper bounds on conference key propagation from graph topology: complex bottleneck structures}
\label{SecUpBnd}

Direct optimization using the Nash-Williams--Tutte formula (\ref{EqNashTutte}) for calculating the conference key propagation rate requires exponential time in the number of parties $N$ because the number of vertex partitions is exponential. There exists an efficient algorithm solving this problem \cite{Barahona1995}, but it has no simple analytic expression. For practice it is often useful to have simple bounds for the quantity of interest. While a lower bound can be found by finding an arbitrary (generally, suboptimal) spanning-tree packing, the Nash-Williams--Tutte formula gives upper bounds. Namely, it follows that
{}
\begin{equation}
\label{EqRateUpperP}
    r_{\rm conf}
    \leq
    \frac{1}{|P|-1}
    \sum_{e\in\mathcal E(P)}r_e
\end{equation}
for any vertex partition $P$. Consideration of the partition into single-element subsets (we will refer to it as the finest partition) $\mathcal V_i=\{i\}$, $i=1,\ldots,N$, gives

\begin{equation}
\label{EqUpperGood}
    r_{\rm conf}\leq\frac{1}{N-1}\sum_{e\in\mathcal E}r_e.
\end{equation}
That is, if we take an arbitrary vertex partition $P$, for example, the finest partition, then Eq.~(\ref{EqRateUpperP}) or (\ref{EqUpperGood}) gives a simple upper bound. In the rest of this subsection, we discuss how to ``guess'' a vertex partition $P$ which gives the tight bound or at least a good one, and consider examples depicted on Figs.~\ref{FigContracSimple} and \ref{FigContracTri}.

Bound (\ref{EqRateUpperP})  also follows from simple arguments based on Eq.~(\ref{EqRateMultigraph}): $k$ edge-disjoint spanning trees have $k(N-1)$ edges in total. This number cannot exceed the total number $n|{\mathcal E}|=n\sum_{e\in\mathcal E}r_e$ of edges in the multigraph $(\mathcal V,{\mathcal E},\{nr_e\})$. A generalization of these arguments to an arbitrary partition $P$ gives Eq.~(\ref{EqRateUpperP}).

From these arguments, it follows that, if the equality in Eq.~(\ref{EqUpperGood}) holds, i.e.,
\begin{equation}
\label{EqRateGood}
    r_{\rm conf}=\frac{1}{N-1}\sum_{e\in\mathcal E}r_e,
\end{equation}
then all edge capacities (bipartite QKD links) are fully used and there are no ``leftover'' edges. The first three rows of Fig.~\ref{FigPack} give examples of this case.  The bottom row in Fig~\ref{FigPack} is an example of the inverse case. It is obvious that the conference key rate cannot exceed one because the conference key rate cannot exceed the total number of bipartite secret key bits of an arbitrary participant. The figure shows a decomposition of this graph into one spanning tree and a residual edge $(1,3)$. This corresponds to the case that the participants do not use the key $K_{(1,3)}$ in the generation of the conference key. 

The existence of ``leftover'' QKD links so that Eq.~(\ref{EqUpperGood}) is strict, or, in other words, if the minimum in the Nash-Williams--Tutte formula (\ref{EqNashTutte}) is achieved not by the finest partition $\mathcal V_i=\{i\}$, $i=1,\ldots,N$, this indicates  the existence of bottleneck structures in our network, which restrict the conference key rate. The minimum in the last example of Fig.~\ref{FigPack} is achieved by the vertex bipartition into the subsets $\mathcal V_1=\{1,2,3\}$ and $\mathcal V_2=\{4\}$.

Figure~\ref{FigContracSimple} gives another interpretation of Eq.~(\ref{EqRateUpperP}). Instead of partitions, one can think about a \textit{contraction} of the graph \cite{Diestel}. Namely, for a given partition $P=\{\mathcal V_\alpha\}$ consider the graph where the sets $\mathcal V_\alpha$ are vertices and two vertices $\mathcal V_\alpha$ and $\mathcal V_\beta$ are connected by an edge if and only if the set $\mathcal E(\mathcal V_\alpha,\mathcal V_\beta)$ of edges between the vertex subsets $\mathcal V_\alpha$ and $\mathcal V_\beta$ is nonempty. This edge obtains the weight
\begin{eqnarray}
    r_{\alpha\beta}=
    \sum_{e\in\mathcal E(\mathcal V_\alpha,\mathcal V_\beta)}
    r_e.
\end{eqnarray}
Then Eq.~(\ref{EqRateUpperP}) is intuitively clear. A contraction of a subset of vertices into one vertex means that the corresponding participants become actually one participant, i.e., they have common knowledge without cost of communication. Another interpretation: we can set the weights of edges (i.e., the bipartite secret key rates) connecting vertices from the same subset to infinity. This situation is more advantageous for conference key agreement, hence, the conference secret key rate cannot decrease under this transformation.

\begin{figure}
    \centering
    \includegraphics[scale=1]{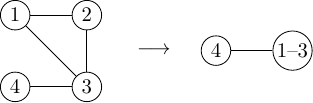}
    \caption{A contraction of the graph from Fig.~\ref{FigPack} (last row) indicating a bottleneck structure, which restricts the conference key rate: connection of node~4 to the rest of the network. The conference key rate in the original graph cannot be larger that the conference (bipartite) key rate in the contracted graph. If the link $(3,4)$ has the rate $r_{(3,4)}=1$, then, for the left graph, $r_{\rm conf}\leq1$.}
    \label{FigContracSimple}
\end{figure}

Generally, consider a bipartition $\mathcal V=\mathcal V_1\sqcup\mathcal V_2$. Then, $|P|=2$ and, according to formula (\ref{EqRateUpperP}), the conference key rate cannot exceed the total capacity of edges between the subsets $\mathcal V_1$ and $\mathcal V_2$. This is an obvious upper bound. The first three lines in the solution (\ref{EqTriSol}) for the triangle network depicted on Fig.~\ref{FigMulti}(a) also reflect bottleneck structures based on bipartitions. For example, if the ``weakest'' part of the network is the connection of vertex~1 to the rest of the network, then the contraction of vertices 2 and 3 gives the tight upper bound. Upper bounds based on bipartitions are used, e.g., in Ref.~\cite{Pirandola2020}. 

However, formula (\ref{EqRateUpperP}) can be used to obtain upper bounds based on more general partitions. An example is given in Fig~\ref{FigContracTri}, where the bottleneck structure is not a connection between two subgraphs, but a triangle. If, as in the previous examples, each edge corresponds to the bipartite key rate 1, then we obtain the upper bound $r_{\rm conf}\leq3/2$ (actually, $r_{\rm conf}=3/2$) originating from a triangle contracted graph, while considerations of only bipartitions give a loose bound $r_{\rm conf}\leq2$.

\begin{figure*}
    \centering
    \includegraphics[scale=1]{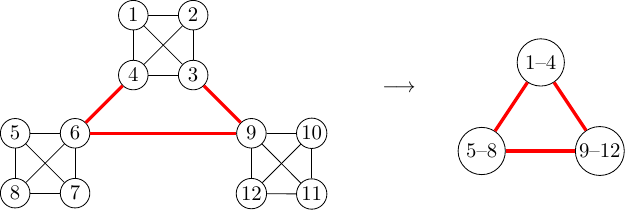}
    \caption{A contraction of a graph indicating a more complex bottleneck structure, than a connection between two subgraphs (like in Fig~\ref{FigContracSimple}): a triangle bottleneck structure. If all bipartite QKD links have the rate $r_e=1$ and, thus, the triangle network with bipartite key rates $r_e=1$ has conference key rate 3/2 (see Fig.~\ref{FigPack}), then, for the left graph, $r_{\rm conf}\leq3/2$.
    }
    \label{FigContracTri}
\end{figure*}

Finally, it is interesting to note that, in order to check that the minimum in Eq.~(\ref{EqNashTutte}) is achieved by the finest partition, i.e.,  Eq.~(\ref{EqRateGood}) is satisfied, one does not need to check all possible partitions, but only partitions of the form $P_I=\{\mathcal V_1,\ldots,\mathcal V_l,\mathcal V_{l+1}\}$, where $I=\{v_1,\ldots,v_l\}\subsetneq[N]$, $\mathcal V_i=\{v_i\}$, $i=1,\ldots,l$, and $\mathcal V_{l+1}=[N]\backslash I$, see Corollary~\ref{CorGoodCond} in Appendix~\ref{SecProofOpt}. In other words, to ensure Eq.~(\ref{EqRateGood}), it is sufficient to check that
\begin{equation}
\label{EqBottleneckSingle}
    \frac{1}{N-1}
    \sum_{e\in\mathcal E}
    r_e
    \leq
    \frac{1}{|I|}
    \sum_{e\in\mathcal E(I)\cup\mathcal E(I,[N]\backslash I)}
    r_e,
\end{equation}
or, equivalently,
\begin{equation}
\label{EqBottleneckSingle2}
    \frac{1}{N-|I|-1}
    \sum_{e\in\mathcal E([N]\backslash I)}
    r_e
    \leq
    \frac{1}{|I|}
    \sum_{e\in\mathcal E(I)\cup\mathcal E(I,[N]\backslash I)}
    r_e
\end{equation}
for all $I\subsetneq [N]$. Here $\mathcal E(I,I')\subset E$ denotes the subset of edges connecting the vertices from $I$ to the vertices from $I'$. As we show in Appendix~\ref{SecProofOpt}, for each $I$, Eqs.~(\ref{EqBottleneckSingle}) and (\ref{EqBottleneckSingle2}) are either both satisfied or both violated. 

Equation~(\ref{EqBottleneckSingle2}) has the following interpretation. 
The left-hand side of it is an upper bound on the conference key propagation in the subnetwork $[N]\backslash I$, see Eq.~(\ref{EqRateUpperP}). Analogously, the right-hand side of Eq.~(\ref{EqBottleneckSingle2}) is the upper bound on the conference key rate for the network where the vertices from $[N]\backslash I$ are contracted into one vertex (hence, the corresponding graph contains $|I|+1$ vertices). The violation of this inequality means that, even if we contract the set of vertices $[N]\backslash I$ into one vertex, the conference key rate in such simplified network is still upper bounded by the upper bound for the subnetwork $[N]\backslash I$. This indicates that the connections of vertices from $I$ with each other and with the other vertices constitute a bottleneck structure.

Thus, the proved equivalence between the Nash-Williams--Tutte formula and the linear program (\ref{EqLP}) reveals properties of the spanning-tree packing.

\subsection{Network optimization}
\label{SecNetOpt}

Let us use the results from the previous subsection to optimize $r_{\rm conf}$ by adding bipartite links to networks in an optimal way. Consider the  problem depicted in Fig~\ref{FigHex} as an example.

\begin{figure}
\centering

\includegraphics[scale=1]{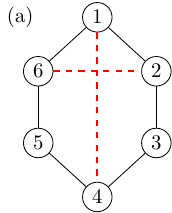}
\qquad\qquad
\includegraphics[scale=1]{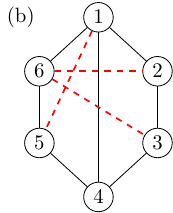}

\caption{Network optimization. All edges here correspond to bipartite secret key rate $r=1$. (a): The initial ringlike network of six nodes gives the conference key rate $r_{\rm conf}=6/5$. If we are allowed to add one link (edge), then which of the two links marked as red-dashed edges is it better to choose? The optimal choice is to add the link $(1,4)$.  (b): If we are  allowed to add one more link, which of the three marked options is optimal? The optimal solution is to add $(2,6)$ or $(3,6)$ because the addition of the link $(1,5)$ leads to a triangle bottleneck structure, see the text.
}
\label{FigHex}
\end{figure}

Initially, we have a ring of six nodes [black solid edges in Fig~\ref{FigHex}(a)]. We assume $r_e=1$ for all edges. Then the initial conference key rate is $r_{\rm conf}=6/5$. Note that the minimum in the Nash-Williams--Tutte formula (\ref{EqNashTutte}) is achieved by the finest partition $\mathcal V=\{1\}\sqcup\ldots\sqcup\{6\}$ and, thus, Eq.~(\ref{EqRateGood}) is satisfied.

Suppose that we are allowed to add one more bipartite QKD link, either $(1,4)$ or $(2,6)$ [red dashed edges in Fig~\ref{FigHex}(a)]. If we add the link $(1,4)$, then the minimum in the Nash-Williams--Tutte formula (\ref{EqNashTutte}) is still achieved by the finest partition and Eq.~(\ref{EqRateGood}) is satisfied, which gives the new conference key rate $r_{\rm conf}=7/5$. 

In contrast, if we choose to add the link $(2,6)$ instead, then the minimum in the Nash-Williams--Tutte formula is achieved by the partition
\begin{equation}
    \mathcal V=\{1,2,6\}\sqcup\{3\}
    \sqcup\{4\}\sqcup\{5\},
\end{equation}
which gives $r_{\rm conf}=4/3$. The conference key rate $4/3$ corresponds to a ring of four nodes (``4-ring''). Thus, the described contraction reveals the bottleneck structure in the form of a 4-ring if the link $(2,6)$ is added.
Adding the link (1,4) leads to a higher rate. 

If we are allowed to add a second additional bipartite QKD link, e.g., $(2,6)$, $(3,6)$, or $(1,5)$ depicted in Fig~\ref{FigHex}(b), then the choices $(2,6)$ and $(3,6)$ [i.e., the links that connect different parts of the graph separated by the central vertical link $(1,4)$] lead to a further increase of the conference key rate to $r_{\rm conf}=8/5$. This again correspond to the finest partition as the optimal one and no bottleneck structures. In contrast, the choice $(1,5)$ leads to the partition
\begin{equation}
    \mathcal V=\{1,4,5,6\}\sqcup\{2\}
    \sqcup\{3\}
\end{equation}
as the optimal one, which corresponds to the contraction of nodes 1, 4, 5, and 6 and, thus, to a triangle bottleneck structure with $r_{\rm conf}=3/2$.

Thus, if we are allowed to allocate a restricted number of bipartite QKD links, bottleneck structures should be avoided.

\section{Security parameter of the STP protocol}
\label{SecSecurityParam}

So far, we assumed the case of perfectly secure bipartite keys. Let us now consider a realistic QKD scenario, where bipartite keys $K_e$, $e\in\mathcal E$, are characterized by the values $\varepsilon_e>0$ of the security parameter, which is based on the trace distance between the real and ideal classical-quantum states after the protocol \cite{Portmann2014,Portmann2022}. Recall that $\varepsilon_e=0$ corresponds to perfect security. Roughly speaking, $\varepsilon_e$ can be associated with the probability that the distributed key is insecure, see Refs.~\cite{Arbekov,Trush2020} for cryptographic operational interpretations.

Consider a tree network like in Fig~\ref{FigTree}. If the number of nodes  in the network is $N$, then any tree contains $N-1$ edges. As we know from Sec.~\ref{SecTreeNet}, the conference key propagation protocol for the tree network can be described as a sequence of one-time-pad encryptions of a chosen bipartite key using the other bipartite keys. During such process, the security parameters of all keys are summed up \cite{BenOrEtAl,Portmann2014,Portmann2022}, hence 
the security parameter of the conference key is
\begin{eqnarray}
    \varepsilon_{\rm conf}=\sum_{e\in\mathcal E}\varepsilon_e.
\end{eqnarray}
If all $\varepsilon_e$ are equal to  $\varepsilon$, then $\varepsilon_{\rm conf}=(N-1)\varepsilon$.

Consider now the general case and formula (\ref{EqRateMultigraph}): a conference key is generated using $k$ spanning trees $T_\beta$, $\beta=1,\ldots,k$. Then, if we form separate conference keys from each spanning tree, then their security parameters are
\begin{equation}
    \varepsilon_{\rm conf}^{(\beta)}=\sum_{e\in T_\beta}\varepsilon_e.
\end{equation}
If we merge all these spanning trees into one long conference key, then we need to sum all the security parameters given above:
\begin{equation}
    \varepsilon=
    \sum_{\beta=1}^k
    \varepsilon_{\rm conf}^{(\beta)}=
    \sum_{\beta=1}^k
    \sum_{e\in T_\beta}\varepsilon_e.
\end{equation}
If all $\varepsilon_e$ are equal to the same number $\varepsilon$, then $\varepsilon^{(\beta)}_{\rm conf}=(N-1)\varepsilon$ and $\varepsilon_{\rm conf}=k(N-1)\varepsilon$. Thus, increasing $k$ and $n$ in Eq.~(\ref{EqRateMultigraph}) allows one to achieve a 
higher asymptotic conference key rate, but decreases its security parameter if we merge the keys obtained from all spanning trees into one long key.

\section{Conclusions and open problems}

We have proposed an optimal solution for the conference key agreement in an arbitrary network of bipartite QKD links. This solution is based on an optimal spanning-tree packing in a graph of bipartite keys. The main tool for proving optimality is the proof of the equivalence between the graph-theoretic Nash-Williams--Tutte formula for the spanning-tree-packing number and a linear program originating from information-theoretic considerations of Csisz\'{a}r and Narayan.

We have shown that the optimality of the STP protocol and the Nash-Williams--Tutte formula can be used for bounding the conference key rate from above and for optimization of QKD link allocations in the network design. Revealing bottlenecks of arbitrary topologies (i.e., originating from arbitrary vertex partitions), not just ``linear'' bottleneck structures (i.e., originating from bipartitions) play a crucial role here.

We highlight three open problems. First,
we assumed that the bipartite key rates $r_e$ are given. However, in practice, not all bipartite QKD links may operate at full capacities simultaneously due to a restricted number of QKD devices. 
This leads to different optimization problems, which include the optimization over the usage of the QKD links (i.e., over $r_e$) with given link capacities and node constraints. Such problems warrant further investigation.

The second open problem concerns the  conference key generation in a subnetwork, where only a subset $\mathcal V'\subset\mathcal V$ of users want to establish a conference key, but other users are trusted and can assist them. Then, instead of spanning trees, we need to consider so called Steiner trees, i.e., trees that contain a given set of vertices (but also may include other vertices). Steiner tree packing for the relative problem of Greenberger-Horne-Zeilinger (GHZ) state distillation in a network of Bell pairs was considered in Refs.~\cite{BaumlAzuma2017,BaumlAzuma2020,BaumlAzuma2021}. Unfortunately, in contrast to the spanning-tree packing, the Steiner tree packing is an NP-hard problem. Nevertheless, does the optimal Steiner tree packing also provide optimal conference key agreement in a subnetwork among all possible protocols?

Finally, it is interesting to generalize these results to the case where genuine multipartite QKD (e.g., based on GHZ states) \cite{QCKAreview} is available between some of the nodes. For example, there are not only bipartite, but also tripartite QKD links in the network. This leads to a transition from the concept of graphs to the concept of hypergraphs. The notions of spanning tree and the spanning-tree packing are naturally generalized for hypergraphs. An open question is its optimality in this more general case.

\section*{Acknowledgments}

This work was funded by the Federal Ministry of Research, Technology and Space BMFTR (Project QuKuK, Grant No.~16KIS1618K).

\appendix

\section{A simple method of finding optimal spanning-tree packing}
\label{SecAlg}

Here we present a method of finding optimal spanning-tree packing for integer edge weights $r_e=L_e$. Formally, this algorithm is not polynomial since it requires checking Eq.~(\ref{EqBottleneckSingle}) for all $I\subsetneq [N]$. Moreover, some steps of this protocol are defined heuristically, nonrigorously. Nevertheless, this simple method allowed us to find optimal spanning-tree packings for the examples considered in this paper without addressing to the complicated algorithm from Ref.~\cite{Barahona1995} and, probably, can be useful also for readers wishing to find spanning-tree packings for small graphs or graphs with recognizable structures (certain patterns of edges, recognizable clusters, etc.).

Consider first the case of integer edge wedge $r_e=L_e$ and no bottlenecks in the network, i.e., Eq.~(\ref{EqRateGood}) is satisfied, or, equivalently, Eq.~(\ref{EqBottleneckSingle}) is satisfied for all $I\subsetneq [N]$. Then we can generate $L_{\rm conf}=\sum_e L_e$ conference key bits in $N-1$ rounds.

The algorithm works then as follows. 

\begin{algorithm}[H]
\caption{Basic algorithm}
\begin{algorithmic}[1]
\Require Graph with integer edge weights $L_e$ satisfying Eq.~(\ref{EqRateGood})
\Ensure Optimal spanning-tree packing
\State $L_{\rm conf} \gets \sum_eL_e$ 
\Comment{Conference key length}
\State $L'_e \gets (N-1)L_e$ for all edges $e$ \Comment{New edge weights}
\For{$\alpha\in\{1,\ldots,L_{\rm conf}-2\}$}
\State Choose a spanning tree $T_\alpha$ with edges of maximal weights $L'_e$
\State $L'_e \gets L'_e-1$ for all $e\in T_\alpha$
\EndFor
\State Choose the last two spanning trees $T_{L_{\rm conf}-1}$ and $T_{L_{\rm conf}}$
\end{algorithmic}
\end{algorithm}

Let us comment the (heuristic) rules of choice of the spanning trees. In each step, we need to choose a spanning tree with edges of possibly maximal weights $L'_e$. Of course, in some cases it is impossible, e.g., when there are $N-1$ edges of the maximal weight, but they do not form a spanning tree. However, generally, the algorithm prescribes to prefer edges with higher weights. The choices of the first $L_{\rm conf}-2$ trees can be arbitrary satisfying this rule. However, the choice of the last but one spanning tree is trickier because, after this choice and the corresponding reduction of the edge weights, the remaining edges of nonzero weight may not form a spanning tree. Hence, the choice of the spanning tree $T_{L_{\rm conf}-1}$ must ensure that the remaining graph is a spanning tree, which will be the last spanning tree $T_{L_{\rm conf}}$. An example of the work of this algorithm is given in Fig.~\ref{FigAlg} (the first and second lines).

Consider now the case when Eq.~(\ref{EqRateGood}) may be not satisfied, or, equivalently, Eqs.~(\ref{EqBottleneckSingle}) may be violated. Then the method is to break the graph into subgraphs until these conditions are satisfied and Algorithm~1 can be applied. Let us rewrite condition (\ref{EqBottleneckSingle}) for an arbitrary graph $(\mathcal V',\mathcal E')$ and integer weights $L_e$:

\begin{equation}
\label{EqBottleneckSingleGen}
    \frac{1}{|\mathcal V'|-1}
    \sum_{e\in\mathcal E'}
    L_e
    \leq
    \frac{1}{|I|}
    \sum_{e\in\mathcal E'(I)\cup\mathcal E'(I,\mathcal V'\backslash I)}
    L_e,
\end{equation}
where the notations $\mathcal E'(I)$ and $\mathcal E'(I,\mathcal V'\backslash I)$ for $I\subset\mathcal V'$ are defined analogously to $\mathcal E(I)$ and $\mathcal E(I,\mathcal V'\backslash I)$. We noticed [see the text after Eq.~(\ref{EqBottleneckSingle2})] that violation of Eq.~(\ref{EqBottleneckSingleGen}) for some $I$ means that the graph with $|I|+1$ vertices where vertices from $\mathcal V'\backslash I$ are contracted into one is a bottleneck structure. This observation suggests to consider two spanning-tree packing problems separately: for the subgraph induced by the vertices $\mathcal V'\backslash I$ and for the aforementioned contracted graph, and then merge the spanning trees. Such splitting of the problem into two problems for smaller graphs works recursively as follows:
\begin{algorithm}[H]
\caption{General algorithm}
\begin{algorithmic}[1]
\Require  Graph $(\mathcal V',\mathcal E')$ with integer edge weights $L_e$
\Ensure Optimal spanning-tree packing
\If{Ineq.~(\ref{EqBottleneckSingleGen}) is satisfied for all  $I\subsetneq\mathcal V'$}
\State Run Basic algorithm for the graph $(\mathcal V',\mathcal E')$
\ElsIf{Ineq.~(\ref{EqBottleneckSingleGen}) is violated for some $I\subsetneq\mathcal V'$}
\State Run General algorithm for the graph $(\mathcal V'',\mathcal E'')$ obtained by the contraction of the vertices from $\mathcal V'\backslash I$ into one vertex
\State Run General algorithm for the subgraph induced by the vertices $\mathcal V'\backslash I$
\State Merge the spanning trees of the graphs from steps 4 and 5 according to the edge correspondence between the input graph $(\mathcal V',\mathcal E')$ and its contraction $(\mathcal V'',\mathcal E'')$
\EndIf
\end{algorithmic}
\end{algorithm}

An example of the work of this algorithm is given in Fig.~\ref{FigAlg} (the third line).

The necessity to check Eqs.~(\ref{EqBottleneckSingle}) for all vertex subsets $I$ and, moreover, further checks of Eqs.~(\ref{EqBottleneckSingleGen}) for subgraphs, make this algorithm inefficient for large graphs. However, it can be applied if the graph contains well-recognizable clusters so that finding $I$ violating Eqs.~(\ref{EqBottleneckSingle}) and (\ref{EqBottleneckSingleGen}) is obvious, or, vice versa, when we can guess that there are no such clusters and Eq.~(\ref{EqRateGood}) is true so that Basic algorithm can be directly applied. Otherwise, one can modify the algorithm to find suboptimal spanning-tree packings. 

So, we can hope that the algorithm can be applied to practical networks as well. There are well-known examples in computer science when formally nonpolynomial algorithms (e.g., the simplex algorithm in linear programming) work well for most practical problems.

\begin{figure*}
%\centering
\includegraphics[scale=1]{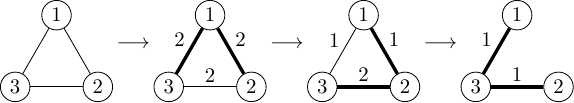}

\vspace{10mm}

\includegraphics[scale=1]{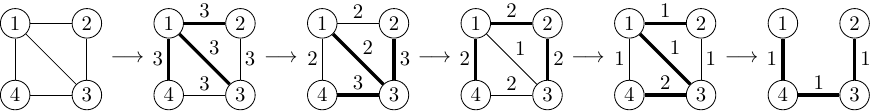}

\vspace{10mm}

\includegraphics[scale=1]{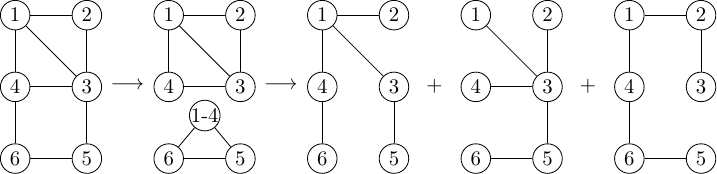}

    \caption{Illustration of the Basic algorithm (first two lines) and the General algorithm (the last line) of finding an optimal spanning-tree packing. All edge weights are equal to one. For the graphs in the first two lines, condition (\ref{EqRateGood}) is satisfied and we know that we can generate a certain number of conference key bits in $N-1$ round ($N-1$ is the number of vertices). We assign new weights $N-1$ to each edge and then, in each step, choose a spanning tree with possibly maximal weights of edges (bold edges). Each use of the edge reduce its weight by one.  The graph in the third line violates Eq.~(\ref{EqBottleneckSingle}) for $I=\{5,6\}$, so, the conference key rate cannot exceed that for the graph where the rest vertices $\{1,2,3,4\}$ are contracted into one, i.e., the triangle graph, or 3/2. Then we solve the spanning-tree packing-problem for both the contracted graph and the subgraph induced by the vertices $\{1,2,3,4\}$ and merge the spanning trees of these graphs according to the edge correspondence  between the original graph and the contracted one. We obtained three spanning trees with not more than two uses of each edge, which gives the rate 3/2. Choices of other spanning trees lead to other optimal solutions.
    }
    \label{FigAlg}
\end{figure*}

Finally, in the case of noninteger edge weights $r_e$, two ways can be considered. The first one is to consider $n$ rounds such that all $nr_e$ are integer.  
The second (heuristic) one is to solve the problem for all weights equal to one (or integers reproducing the proportions between $r_e$ with some precision) and, with the obtained list of spanning trees, solve the linear program (\ref{EqRateGraph}) (see below an equivalent formulation of the optimal spanning-tree-packing problem) optimizing their weights. Note that Eq.~(\ref{EqRateGraph})  is indeed a linear program for $w_\alpha$ if the set of the trees is fixed. The requirement that $w_\alpha$ must be rational is not a restriction since, if $r_e$ are rational, there exists an optimal solution of the linear program with rational $w_\alpha$.

\section{Proof of optimality of the spanning-tree-packing conference key propagation}
\label{SecProofOpt}

In this section we prove Theorem~\ref{ThOpt} and Proposition~\ref{ThLP}.

\subsection{Reduction to a linear program for a multigraph network}
\label{SecReduction}

The basic tool is the information-theoretic result by Czisz\'{a}r and Narayan \cite{Csiszar2004} about the classical conference secret key capacity. One of the models considered in their paper is as follows. 
In each repetition (round),
iid $N$-tuples of random variables $(X_1,\ldots,X_N)$ are delivered to the users. That is, the user $i$ observes the random variable $X_i$. The random variables $X_i$ for different $i$ are generally dependent according to a known joint distribution for $(X_1,\ldots,X_N)$, but the $N$-tuples for different rounds are independent.
The participants communicate over a noiseless public channel. The eavesdropper has no information about $X_i$, but has access to the public communication. The participants want to agree on a common conference secret key with the highest possible rate. See Ref.~\cite{Csiszar2004} for formal definitions, but they match our definitions of the conference key capacity (\ref{EqConfCapacity}). Then the conference secret key capacity is given by
{}
\begin{equation}
\label{EqCN}
\begin{split}
    &Z=H(X_{[N]})-R_{\rm CO},
    \\
    &R_{\rm CO}=
    \min_{R_1,\ldots,R_N}
    \sum_{i=1}^N R_i
    \quad
    \text{such that}
    \\
    &\sum_{i\in I} R_i
    \geq
    H(X_I|X_{[N]\backslash I}),
    \quad \forall I\subsetneq [N],
\end{split}
\end{equation}
where $X_I=(X_i)_{i\in I}$, $H$ denotes the Shannon entropy, $R_i$ is the amount of bits of information per round announced by the $i$th participant, and $R_{\rm CO}$ is the smallest achievable rate of ``communication for omniscience,'' i.e., the minimal total amount of bits of public communication (per round) which allows all participants to recover all iid repetitions of all random variables $X_1,\ldots,X_N$. The information-theoretic meaning of these relations and constraints is given after Theorem~\ref{ThOpt}.

The definition of $X_i$ for our case is described in Sec.~\ref{SecOpt}. Each participant $i$ obtains $\sum_{e\in\mathcal E_i}r_{e}$ bipartite perfectly secure bits per round, where $\mathcal E_i$ is the set of edges incident to the vertex $i$. This bits constitute $X_i$. Of course, different $X_i$ are not independent since, for each edge $e=(i,j)$, $r_e$ bits are included into both $X_i$ and $X_j$. Thus,
\begin{equation}
\begin{split}
    H(X_{[N]})&=\sum_{e\in\mathcal E} r_{e},
    \\
    H(X_I|X_{[N]\backslash I})
    &=\sum_{e\in\mathcal E(I)}
    r_{e},
\end{split}
\end{equation}
from which we obtain the linear program (\ref{EqLP}).
Thus, the proof of Theorem~\ref{ThOpt} has been reduced to Proposition~\ref{ThOpt}, which we prove in the next subsection.

\subsection{Proof of Proposition~\ref{ThLP}}

\begin{lemma}
\label{LemCUpper}
Consider an arbitrary partition $P=\{\mathcal V_1,\ldots,\mathcal V_p\}$ of the vertex set $\mathcal V$ into nonempty subsets, i.e., $\mathcal V=\mathcal V_1\sqcup\ldots\sqcup\mathcal V_p$, $p\equiv |P|\geq2$. Then the following upper bound for $Z$ form linear program (\ref{EqLP}) holds:
\begin{equation}
    \label{EqCUpperP}
    Z\leq\frac{1}{|P|-1}
    \sum_{e\in\mathcal E(P)}r_e,
\end{equation}
where, recall, $\mathcal E(P)$ is the set of cross edges in the graph, i.e., the edges whose vertices lie in different partition subsets.
\end{lemma}

Note that this observation in a different form was already made in Ref.~\cite{Csiszar2004}. For self-consistency, we provide the proof also here.

\begin{proof}
It will be convenient to introduce the notation
\begin{equation}
    r[\mathcal E']=\sum_{e\in\mathcal E'}r_e
\end{equation}
for an arbitrary subset $\mathcal E'\subset\mathcal E$.
The sum of the constraints (\ref{EqLPCons}) for $I_\alpha=[N]\backslash \mathcal V_\alpha$ for all $\alpha=1,\ldots,|P|$ gives
\begin{multline}
    \label{EqRCObnd}
    (|P|-1)\sum_{i=1}^N R_i
    \geq 
    \sum_{\alpha=1}^{|P|}
    r[{\mathcal E}([N]\backslash\mathcal V_\alpha)]
    \\
    =\sum_{\alpha=1}^{|P|}
    \Big(
    r[{\mathcal E}]
    -
    r[{\mathcal E}(\mathcal V_\alpha)]
    -
    r[{\mathcal E}
    (\mathcal V_\alpha,[N]\backslash\mathcal V_\alpha)|
    \Big)
    \\
    =(|P|-1)r[{\mathcal E}]
    -r[{\mathcal E}(P)],
\end{multline}
where, recall, ${\mathcal E}(I,I')\subset{\mathcal E}$ denotes the subset of edges connecting the vertices from $I$ to the vertices from $I'$. Here we have used that each $R_i$ participates in exactly $|P|-1$ constraints for various $\alpha$,
\begin{equation}
    \sum_\alpha
    r[{\mathcal E}(\mathcal V_\alpha,[N]\backslash \mathcal V_\alpha)]
    =
    2r[{\mathcal E}(P)],
\end{equation}
(since each cross edge is counted twice), and
\begin{equation}
    \sum_\alpha
    r[{\mathcal E}(\mathcal V_\alpha)]
    +
    r[{\mathcal E}(P)]
    =
    r[{\mathcal E}].
\end{equation}
Substitution of the lower bound for $R_{\rm CO}$ from Eq.~(\ref{EqRCObnd}) to Eq.~(\ref{EqLPC}) gives Eq.~(\ref{EqCUpperP}).
\end{proof}

Lemma~\ref{LemCUpper} and the tree packing (Nash-Williams--Tutte) theorem allows us to establish a relation between the optimal solution of the linear program and spanning-tree packings in the multigraph.

\begin{corollary}
\label{CorLPNW}
    The optimal solution of the linear program (\ref{EqLP}) gives the value
    \begin{equation}
    \label{EqCTreePackTh}
        Z=\min_{
        \begin{smallmatrix}
            {\rm vertex}\\
            {\rm partitions}\:P
        \end{smallmatrix}
        }
        \frac{1}{|P|-1}
        \sum_{e\in\mathcal E(P)}r_e,
    \end{equation}
    which coincides with the Nash-Williams--Tutte formula (\ref{EqNashTutte}). 
\end{corollary}

\begin{proof}
    By Lemma~\ref{LemCUpper}, $Z$ cannot exceed the right-hand side of Eq.~(\ref{EqCTreePackTh}). From the other side, the the right-hand side of Eq.~(\ref{EqCTreePackTh}) is achievable by the spanning-tree-packing protocol precisely in view of the Nash-Williams--Tutte formula (\ref{EqNashTutte}). Note that it is valid also for real $r_e$. Indeed, for $n$ rounds, we can consider the multigraph $(\mathcal V,\mathcal E,\{\lfloor nr_e\rfloor\})$. Its spanning-tree-packing number $L_{\rm conf}^{(n)}$ is given by  Eq.~(\ref{EqNashTutten}) with $nr_e$ replaced by $\lfloor nr_e\rfloor$. The limit of $L_{\rm conf}^{(n)}/n$ as $n\to\infty$ gives again Eq.~(\ref{EqNashTutte}).   
    Hence, the right-hand side of Eq.~(\ref{EqCTreePackTh}) is also (asymptotically) achievable and, hence, optimal.
\end{proof}
This finishes the proof Proposition~\ref{ThLP} and Theorem~\ref{ThOpt}.

\begin{remark}
    In the proof of Corollary~\ref{CorLPNW}, we implicitly assumed that the conference key propagation protocol based on the spanning-tree-packing problem (\ref{EqRateGraph}) gives a feasible solution to the linear program. It follows from the results of Csisz\'{a}r and Narayan: Any feasible algorithm of the conference key generation from bipartite secret keys must satisfy the restrictions of  linear program (\ref{EqLP}). 
    However, it is instructive to show it explicitly, which will be done in the next subsection.
\end{remark}

Finally, let us consider the case when the minimum in Nash-Williams--Tutte formula (\ref{EqNashTutte}) is achieved by the finest partition $P_{\rm finest}=\{\{1\},\ldots,\{N\}\}$, i.e.,  Eq.~(\ref{EqRateGood}) is satisfied. According the Nash-Williams--Tutte formula, we need to prove that 
\begin{equation}
\label{EqGoodCondPre}
    \frac{r[\mathcal E]}{N-1}
    \leq
    \frac{r[\mathcal E(P)]}{|P|-1}
\end{equation}
for all partitions $P$. However, it turns out that it is sufficient to check Eq.~(\ref{EqGoodCondPre}) only for partitions of the form $P_I=\{\mathcal V_1,\ldots,\mathcal V_l,\mathcal V_{l+1}\}$, where $I=\{v_1,\ldots,v_l\}\subsetneq[N]$, $\mathcal V_i=\{v_i\}$, $i=1,\ldots,l$, and $\mathcal V_{l+1}=[N]\backslash I$. Application of Eq.~(\ref{EqGoodCondPre}) to this particular form of a partition gives
\begin{equation}
\label{EqGoodCond}
    \frac{r[\mathcal E]}{N-1}
    \leq
    \frac{r[{\mathcal E}(I)]
    +
    r[{\mathcal E}(I,[N]\backslash I)]}
    {|I|}.
\end{equation}
We can prove the following:

\begin{proposition}
\label{CorGoodCond}
Equation~(\ref{EqGoodCondPre}) is true for all vertex partitions $P$ iff Eq.~(\ref{EqGoodCondPre}) is true for all partitions $P_I$ of the form given above (i.e., for all nonempty $I\subsetneq[N]$).
\end{proposition}

\begin{proof}
In one direction, the statement is obvious: since Eq.~(\ref{EqGoodCondPre}) is a particular case of Eq.~(\ref{EqGoodCondPre}), then, if Eq.~(\ref{EqGoodCondPre}) is satisfied for all vertex partitions $P$, then  Eq.~(\ref{EqGoodCondPre}) is satisfied for all partitions $P_I$. Let us prove the other direction.

Put by definition
    \begin{equation}
        \label{EqGoodRi}
        R_i=r[\mathcal E_i]-\frac{r[\mathcal E]}{N-1}.
    \end{equation}
Substitution of Eq.~(\ref{EqGoodRi}) into Eq.~(\ref{EqLPCons}) for an arbitrary $I$ gives
\begin{equation}
    2r[\mathcal E(I)]
    +r[\mathcal E(I,[N]\backslash I)]-
    \frac{|I|}{N-1}
    r[\mathcal E]
    \geq r[\mathcal E(I)],
\end{equation}
from which
Eq.~(\ref{EqGoodCond}) follows. Hence, if all Eq.~(\ref{EqGoodCond}), we have an explicit solution satisfying all constraints (\ref{EqGoodCondPre}) and, as it follows from the proved results and also can be checked explicitly, gives $Z=r[\mathcal E]/(N-1)$.

\end{proof}

Note that Eqs.~(\ref{EqGoodCond}) are trivially satisfied as equalities for $I=[N]\backslash\{k\}$, $k=1,\ldots,N$. Actually, Eq.~(\ref{EqGoodRi}) is a solution of the the system of $N$ linear equations obtained by replacement of ``$\geq$'' by ``$=$'' in constraints (\ref{EqLPCons}) for $I=[N]\backslash\{k\}$, $k=1,\ldots,N$. We obtained an explicit form of $R_i$ yielding the optimal value for the linear program if there are no bottleneck structures in the network.

In the main text, we also use the equivalent form of Eq.~(\ref{EqGoodCond}):
\begin{equation}
\label{EqGoodCondAlt}
    \frac{r[\mathcal E([N]\backslash I)]}{N-|I|-1}
    \leq
    \frac{r[\mathcal E(I)]+r[\mathcal E(I,[N]\backslash I)]}{|I|}.
\end{equation}
Indeed, Eq.~(\ref{EqGoodCond}) can be rewritten as
\begin{multline}
    \frac{r[\mathcal E(I)]+
    r[\mathcal E(I,[N]\backslash I)]
    +r[\mathcal E([N]\backslash I)]}{N-1}
    \\
    \leq 
    \frac{r[\mathcal E(I)]+r[\mathcal E([N]\backslash I)]}{|I|},
\end{multline}
from which Eq.~(\ref{EqGoodCondAlt}) follows.

\subsection{Spanning-tree-packing protocol satisfies the information-theoretic constraints}

Let us show that the conference key rate given by the Nash-Williams--Tutte formula (\ref{EqNashTutte}) does not exceed $C$ from (\ref{EqLP}), i.e., the information-theoretic constraints (\ref{EqLPCons}) are satisfied. In principle, we do not need to prove this because it is a direct consequence of Csisz\'{a}r and Narayan's result: the rate of any conference key propagation algorithm cannot exceed the information-theoretic bound (\ref{EqLP}). However, it is instructive to give a direct and constructive proof.

We will use  another (equivalent) formulation of the spanning-tree-packing problem~\cite{Barahona1995}. For a given graph $(\mathcal V,\mathcal E)$ and ``edge capacities'' (in our case -- bipartite key rates) $\{r_e\}_{e\in\mathcal E}$, we need to choose a finite set of spanning trees $\{T_\alpha\}$ of this graph  and assign rational weights $\{w_\alpha\geq0\}$ to them such that

\begin{equation}
\label{EqRateGraph}
\begin{split}
    &\sum_\alpha w_\alpha=r_{\rm conf}\to\max_{\{w_\alpha\}},
    \\
    &\sum_{\alpha\colon T_\alpha\ni e}
    w_\alpha \leq r_e,\quad \text{for all }e\in\mathcal E,
\end{split}
\end{equation}
where $T_\alpha\ni e$ means that the spanning tree $T_\alpha$ contains the edge $e$. The constraints in Eq.~(\ref{EqRateGraph}) mean that each bipartite secret bit from each $r_e$ can be used only in one spanning tree. In other words, if we generate one conference bit using the tree $T_\alpha$, we spend one bit from each of the edges constituting this tree. The weight $w_\alpha$ is the average number of usages of the tree $T_\alpha$ per round. For example, the spanning-tree packings depicted in the first line of Fig.~\ref{FigPack}, correspond to $\alpha\in\{1,2,3\}$ (indices of three spanning trees) and $w_\alpha=1/2$ for each $
\alpha$: each spanning tree is used once per two rounds. In the second and the third lines of Fig.~\ref{FigPack}, we have $w_\alpha=1/3$ for each of the five spanning trees and $w_\alpha=1$ for each of the two spanning trees, respectively. The constraints in Eq.~(\ref{EqRateGraph}) mean that, for each edge, the total number of usages of the conference bits cannot exceed its ``capacity'' (number of bipartite bits generated per round) $r_e$. 

Since, as we mentioned in Sec.~\ref{SecSTP}, we use a slightly different formulation of the spanning-tree-packing problem in comparison with the standard one, let us prove the equivalence of Eqs.~(\ref{EqRateMultigraph}) and (\ref{EqRateGraph}).

\begin{observation}
The optimization problems (\ref{EqRateMultigraph}) and (\ref{EqRateGraph}) are equivalent.
\end{observation}

\begin{proof}
Indeed, from one side, consider a set of spanning trees $\{\widetilde T_\beta\}_{\beta=1}^k$ in the multigraph $(\mathcal V,{\mathcal E},\{nr_e\})$. We can obviously map them into a set of spanning trees in the graph $(\mathcal V,\mathcal E)$ by just ignoring the edge multiplicities. However, different spanning trees in the multigraph can correspond to the same spanning tree in the graph $(\mathcal V,\mathcal E)$, see Fig.~\ref{FigMulti2One}.  Denote this map as $f$. Thus, $f$ is in general a many-to-one correspondence. Denote $f^{-1}(T_\alpha)$ the preimage of $T_\alpha$, i.e., the set of the spanning trees from the multigraph $(\mathcal V,\mathcal E,\{ nr_e\})$ mapped into the spanning tree $T_\alpha$ of the graph $(\mathcal V,{\mathcal E})$. Then, put
\begin{equation}
w_\alpha=\frac{f^{-1}(T_\alpha)}n
\end{equation}
and
\begin{equation}
\sum_\alpha w_\alpha = \frac kn.
\end{equation}
Thus, a feasible solution of problem (\ref{EqRateMultigraph}) is a feasible solution of problem (\ref{EqRateGraph}) with the same value of the conference key rate. 

\begin{figure}
    \centering
    \includegraphics[scale=1]{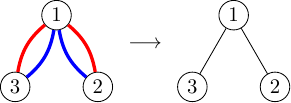}
    \caption{If we ignore the edge multiplicities of a multigraph, several spanning trees of it (here: two spanning trees of the multigraph, which are depicted by red and blue) can map into the same spanning tree of the resulting graph.}
    \label{FigMulti2One}
\end{figure}

From the other side, a feasible solution $\{T_\alpha,w_\alpha\}$ of the problem (\ref{EqRateGraph}) with rational $w_\alpha$ can be mapped to a feasible solution of problem (\ref{EqRateMultigraph}) in the same manner.  
Namely, there exists $n$ such that all $nw_\alpha$ are integer. Then, by construction, we have found $k=n\sum_\alpha w_\alpha$ edge-disjoint spanning trees in the multigraph $(\mathcal V,\mathcal E,\{nr_e\})$. 

Hence, problems (\ref{EqRateMultigraph}) and (\ref{EqRateGraph}) are equivalent.
\end{proof}

We are ready to prove the main result of this subsection.
\begin{lemma}
\label{LemTrees2LP}
    The conference key rate $r_{\rm conf}$ given by (\ref{EqRateMultigraph}) can be expressed as
    \begin{equation}
    \label{EqTreeFromLP}
    r_{\rm conf}=\sum_{e\in\mathcal E}r_e
        -\sum_{i=1}^N R_i,
    \end{equation}
where $R_i$ satisfy constraints (\ref{EqLPCons}).
\end{lemma}

According to the Csisz\'{a}r-Narayan construction, $R_i$ have the meaning of the amount of information sent by the $i$th participant (per round) to the public channel in the spanning-tree-packing algorithm. The proof is based on explicit calculation of the information announced by the participants. We will give the corresponding explanations in the proof, though the proof of a formal statement of the lemma does not rely on this interpretation of $R_i$.

\begin{proof}
    
Denote $d^{(\alpha)}_i$ the nodal degree of the vertex $i$ in the $\alpha$th spanning tree. According to the algorithm, for the $\alpha$th spanning tree, the participant $i$, announce $R_i^{(\alpha)}=d^{(\alpha)}_i-1$ bits of information. As we discussed after Eq.~(\ref{EqRateGraph}), the weight $w_\alpha$ of the $\alpha$th spanning tree is a fraction of using this spanning tree, i.e., for large number $n$ of rounds, the $\alpha$th spanning tree is used approximately $w_\alpha n$ times. So, the participant $i$ announces on average $\sum_{\alpha}w_\alpha R_i^{(\alpha)}$ bits per round. 

In addition, as it was mentioned in Sec.~\ref{SecUpBnd}, not always the full capacity $r_e$ of all edges is used (like in the bottom line of Fig.~\ref{FigPack}). Thus, the differences
\begin{equation}
    \sum_{e\ni i}
    \Big(
    r_{e}-
    \sum_{\alpha\colon T_\alpha\ni e} w_\alpha 
    \Big)
\end{equation}
for all edges $e$ can be nonzero. In other words, in general, there are constraints in the spanning-tree-packing problem (\ref{EqRateGraph}) satisfied as strict inequalities. 

Constraints in (\ref{EqLP}) assume ``omniscience'': participants must be able to recover all bipartite secret bits. In this paradigm, the participants have to announce the bipartite secret bits not used in the conference key propagation protocol. Each ``leftover'' bipartite secret bit is shared by two participants. Let us demand that, on average, one half of this bit is announced by one participant and one half is announced by the other one. That is, in one half rounds this bit is announced by one participant and in one half of rounds -- by the other one.
Then, the total amount of information announced by the $i$th participant per round is
\begin{equation}
\label{EqRi}
    R_i=
    \sum_{\alpha}
    w_\alpha
    (d^{(\alpha)}_i-1)
    +\frac12
    \sum_{e\ni i}
    \Big(
    r_{e}-
    \sum_{\alpha\colon T_\alpha\ni e} w_\alpha 
    \Big).
\end{equation}

We have defined $R_i$ and now start the proof that, with such $R_i$, Eq.~(\ref{EqTreeFromLP}) and constraints (\ref{EqLPCons}) are satisfied. The proof will not use the interpretation of $R_i$ as the amount of information announced by the participants.

For the right-hand side of Eq.~(\ref{EqTreeFromLP}), we have
\begin{equation}
\begin{split}
    \sum_{e\in\mathcal E}r_e-\sum_{i=1}^N R_i
    &=
    \sum_{e\in\mathcal E} r_{e}
    -
    \sum_{i=1}^N
    \sum_{\alpha}
    w_\alpha
    (d^{(\alpha)}_i-1)
    \\
    &-
    \sum_{e\in\mathcal E}
    \Big(
    r_{e}-
    \sum_{\alpha\colon T_\alpha\ni e}
	w_\alpha      
    \Big)\\    
    &=
    \sum_\alpha w_\alpha
    \left(
    \sum_{e\in T_\alpha}
    1
    +N
    -
    \sum_{i=1}^N
    d^{(\alpha)}_i
    \right).
\end{split} 
\end{equation}
Each spanning tree has $N-1$ edges, hence $\sum_{e\in T_\alpha}1=N-1$ for all $\alpha$. Also, the sum of nodal degrees of vertices in a subgraph is the number of edges in this subgraph multiplied by two (since each edge is counted twice), hence $\sum_i d^{(\alpha)}_i=2(N-1)$ and 
\begin{multline}
    \sum_{e\in\mathcal E}r_e-\sum_{i=1}^N R_i
	=   
	\left(\sum_\alpha w_\alpha\right) 
	[N-1+N-2(N-1)]
	\\
	=
	\sum_\alpha w_\alpha=r_{\rm conf}.
\end{multline}
We have proved Eq.~(\ref{EqTreeFromLP})

Now we need to prove that $R_i$ given by Eqs.~(\ref{EqRi})  satisfy constraints (\ref{EqLPCons}). The substitution gives for an arbitrary nonempty $I\subsetneq[N]$:
\begin{equation}
\label{EqLemTreeLPCons}
    \sum_{i\in I}
    R_i
    =
    \sum_{\alpha}
    w_\alpha
    \sum_{i\in I}
    (d^{(\alpha)}_i-1)
    +
    \frac{1}{2}
    \sum_{i\in I}
    \sum_{e\ni i}
    \Big(
    r_{e}-
    \sum_{\alpha\colon T_\alpha\ni e} w_\alpha 
    \Big).
\end{equation}
Now, for each spanning tree, we have
\begin{equation}
\label{EqLemTreeLPCons1}
    \sum_{i\in I} d^{(\alpha)}_i=2|\mathcal E_\alpha(I)|+
    |\mathcal E_\alpha(I,[N]\backslash I)|,
\end{equation}
where $\mathcal E_\alpha(I)$ is the set of edges of the subgraph of the $\alpha$th spanning tree induced by vertices $I$ and $\mathcal E_\alpha(I,[N]\backslash I)$ is the set of edges in the $\alpha$th spanning tree connecting vertices from $I$ to vertices outside $I$. 

For the last sum in Eq.~(\ref{EqLemTreeLPCons}), we have
\begin{eqnarray}
    \frac{1}{2}
    \sum_{i\in I}
    \sum_{e\ni i}
    \Big(
    r_{e}&-&
    \sum_{\alpha\colon T_\alpha\ni e}
    w_\alpha 
    \Big)
    \nonumber
    \\&=&
    \sum_{e\in\mathcal E(I)}
    \Big(
    r_{e}-
    \sum_{\alpha\colon T_\alpha\ni e}
    w_{\alpha}
    \Big)   
    \nonumber
    \\&+&
    \frac{1}{2}
    \sum_{e\in\mathcal E(I,[N]\backslash I)}
    \Big(
    r_{e}-
    \sum_{\alpha\colon T_\alpha\ni e}
    w_{\alpha}
    \Big)
    \nonumber
    \\
    &\geq&
    r[{\mathcal E}(I)]
    -
    \sum_{\alpha}
    w_\alpha
    |\mathcal E_\alpha(I)|.
\label{EqLemTreeLPCons2}
\end{eqnarray}
Substitution of Eq.~(\ref{EqLemTreeLPCons1}) and Eq.~(\ref{EqLemTreeLPCons2}) into Eq.~(\ref{EqLemTreeLPCons}) gives
{}
\begin{equation}
    \sum_{i\in I}R_i
    \geq 
    r[{\mathcal E}(I)]
    +
    \sum_{\alpha}
    w_\alpha
    \Big[|\mathcal E_\alpha(I)|+|\mathcal E_\alpha(I,[N]\backslash I)|-|I|\Big].
\end{equation}

To finish the proof of Eqs.~(\ref{EqLPCons}), it is sufficient to prove that
\begin{equation}
\label{EqTreeSubgraphIneq}
|\mathcal E_\alpha(I)|+|\mathcal E_\alpha(I,[N]\backslash I)|\geq|I|.    
\end{equation}
The left-hand side is the number of edges incident on the vertices from $I$ in the $\alpha$th spanning tree. Consider the subgraph $(I,\mathcal E_{\alpha}(I))$ of the $\alpha$th spanning tree induced by the vertices $I$. It is not necessarily connected: some vertices from $I$ can be connected via vertices outside $I$. 

For the clarity of arguments, consider first the case that the subgraph $(I,\mathcal E_{\alpha}(I))$ is connected. Then it is also a tree and thus contains $|I|-1$ edges. Also the set $\mathcal E_\alpha(I,[N]\backslash I)$ has at least one element, which connects the vertices from $I$ to the vertices outside $I$ (since each spanning tree is a connected graph). Hence, Eq.~(\ref{EqTreeSubgraphIneq}) is true in this case. 

Consider now the general case where the subgraph $(I,\mathcal E_{\alpha}(I))$ is not necessarily connected.  Denote $\{I_\mu\}$ the subsets of $I$ forming connected components of this subgraph. Then, again, each connected component contains is a tree and contains $|I_\mu|-1$ edges. Also, for each $\mu$, there exists at least one edge of the $\alpha$th spanning tree connecting the vertices from $I_\mu$ to the vertices outside $I_\mu$ (and outside $I$ because otherwise such vertices would belong to the same connected component). Thus,
\begin{equation}
|\mathcal E_\alpha(I)|+|\mathcal E_\alpha(I,[N]\backslash I)|\geq
\sum_\mu
[(|I_\mu|-1)+1]
=|I|.    
\end{equation}
This finishes the proof of Eqs.~(\ref{EqTreeSubgraphIneq}), (\ref{EqLPCons}) and thus the whole lemma.
\end{proof}

%\bibliography{main.bib}


%

\end{document}